\documentclass[a4paper,preprintnumbers,floatfix,superscriptaddress,pra,twocolumn,showpacs,notitlepage,longbibliography]{revtex4-2}

\usepackage{amsmath, amssymb, amsthm, bm, physics}
\usepackage{xcolor, graphicx, subfigure, float, hyperref}
\usepackage{tcolorbox}
\usepackage{dsfont}

\newtheorem{thm}{Theorem}[]

\definecolor{forestgreen}{rgb}{0.13, 0.55, 0.13}

\begin{document}
\title{General Method for Classicality Certification in the Prepare and Measure Scenario}

\author{Carlos de Gois}
\affiliation{Instituto  de  Física ``Gleb  Wataghin'', Universidade  Estadual  de  Campinas, CEP 13083-859, Campinas, Brazil}
\author{George Moreno}
\affiliation{International Institute of Physics, Federal University of Rio Grande do Norte, 59070-405 Natal, Brazil}
\author{Ranieri Nery}
\affiliation{International Institute of Physics, Federal University of Rio Grande do Norte, 59070-405 Natal, Brazil}
\author{Samuraí Brito}
\affiliation{International Institute of Physics, Federal University of Rio Grande do Norte, 59070-405 Natal, Brazil}
\author{Rafael Chaves}
\affiliation{International Institute of Physics, Federal University of Rio Grande do Norte, 59070-405 Natal, Brazil}
\affiliation{School of Science and Technology, Federal University of Rio Grande do Norte, 59078-970 Natal, Brazil}
\author{Rafael Rabelo}
\affiliation{Instituto  de  Física ``Gleb  Wataghin'', Universidade  Estadual  de  Campinas, CEP 13083-859, Campinas, Brazil}

\begin{abstract}
   Preparation and measurement of physical systems are the operational building blocks of any physical experiment, and to describe them is the first purpose of any physical theory. It is remarkable that, in some situations, even when only preparation and measurement devices of a single system are present and they are uncharacterized, it is possible to distinguish between the behaviours of quantum and classical systems relying only on observational data. Certifying the physical origin of measurement statistics in the prepare and measure scenario is of primal importance for developing quantum networks, distributing quantum keys and certifying randomness, to mention a few applications, but, surprisingly, no general methods to do so are known. We progress on this problem by crafting a general, sufficient condition to certify that a given set of preparations can only generate classical statistics, for any number of generalized measurements. As an application, we employ the method to demonstrate non-classicality activation in the prepare and measure scenario, also considering its application in random access codes. Following that, we adapt our method to certify, again through a sufficient condition, whether a given set of measurements can never give rise to non-classical behaviors, irrespective of what preparations they may act upon. This, in turn, allows us to find a large set of incompatible measurements that cannot be used to demonstrate non-classicality, thus showing incompatibility is not sufficient for non-classicality in the prepare and measure scenario.
\end{abstract}

\maketitle

\section{Introduction}

    Quantum theory, albeit ubiquitous and extensively tested, still presents us with interesting and unintuitive phenomena even for the simplest physical systems. The paradigmatic examples are Bell nonlocality \cite{bell-epr-1964, brunner-nonlocality-2014} and Einstein-Podolsky-Rosen (EPR) steering \cite{wiseman-steering-2007, cavalcanti-steering-2016, uola-steering-2020} --- manifest as strong correlations between space-like separated experiments performed by independent observers --- that can unambiguously discern between classical and quantum predictions. More recently, a similar division was found to arise in a setup closely related to quantum communication tasks, the so-called prepare and measure (PAM) scenario \cite{gallego-pam-2010}. 
    
    Descriptions of the preparation of physical systems and of their measurements are the building blocks of any physical theory.
    Hence, it is notable that even in a semi-device independent approach --- relying alone on observational data and mild assumptions about the state preparation --- such scenario is already enough to distinguish between quantum and classical behaviors. Apart from its foundational relevance, certifying that the systems employed are indeed quantum and behave as expected is an essential task in applications of the PAM scenario, ranging from communication in quantum networks \cite{bowles2015testing,wang2019characterising} and self-testing \cite{tavakoli2018self,miklin2020universal} to quantum key distribution \cite{pawlowski-qkd-2011}, randomness certification \cite{passaro-randomness-2015} and random access codes \cite{li2012semi}, also figuring at the core of informational principles for quantum theory \cite{pawlowski2009information,chaves2015information} and the modelling of paradigmatic gedanken experiments such as the so-called delayed choice experiment \cite{chaves2018causal}.
    
    It has been long known that Holevo's bound \cite{holevo-bound-1973} limits the amount of information that may be retrieved in such a scenario, implying that quantum messages cannot transmit more information than their classical counterparts. Notwithstanding, the non-classicality of the measurement outcomes can still be witnessed if one imposes certain assumptions on the preparation device \cite{gallego-pam-2010,bowles2014certifying,chaves2015device,van2017semi,tavakoli2020informationally,Poderini2020}. For instance, even though a qubit can transmit at most a bit of information, it can still generate measurement statistics that cannot be reproduced by a classical bit \cite{gallego-pam-2010,bowles2014certifying}. On the more applied side, such dimension assumptions can also be used to derive device-independent witnesses for the Hilbert space dimension of the prepared states \cite{brunner2013dimension,George2020}. 

    Within this context, a task of primal importance is to be able to decide whether a given set of prepared states can or cannot give rise to non-classical behaviours. In a Bell scenario, for instance, it has been long known that entanglement is a necessary but insufficient ingredient for non-locality, as there are entangled states that cannot violate any Bell inequality \cite{werner-quantum-1989, barrett-nonsequential-2002}, giving rise to general methods for deciding whether an entangled state is local \cite{cavalcanti-method-2016,hirsch-method-2016}. More than its fundamental relevance, such criteria can also be put to use to show hidden non-locality or activation of non-locality, the former related to the fact that local states can have their non-locality activated if one performs local filtering prior to measurements \cite{popescu-filtering-1995, hirsch-hidden-2013,gallego2014nonlocality} and the latter related to measurements on a number of copies of such states \cite{navascues-activation-2011, palazuelos-superactivation-2012,cavalcanti2011quantum}. In spite of the PAM relevance, it was only recently that the first test, though of limited applicability, has been proposed in order to detect the non-classicality of quantum state preparations \cite{Poderini2020}. Strikingly, however, as opposed to the Bell case, no general methods are known to provide such certification. That is precisely the problem we solve in this paper.
    
    Inspired by results crafted for constructing local hidden variable models for entangled states \cite{cavalcanti-method-2016, hirsch-method-2016}, we propose a general method allowing to certify whether the behaviors arising from a preparation set are always classically reproducible, irrespective of what generalized measurements are applied on them. Taking the measurements as the resource, we conversely derive a method to test whether a measurement set can be classically reproduced no matter what preparations they act upon. We demonstrate the applicability of both methods in a number of cases, and in particular employ them to show activation of non-classicality of a given set of states, and also prove the insufficiency of measurement incompatibility for observing non-classical behaviors in the prepare and measure scenario.


\section{Prepare and measure scenario}

    In the prepare and measure scenario, a preparation device $P$ takes a random variable $x \in \mathcal{X}$ as input and, subjected to it, prepares a physical system in a specific physical state. A second device, $M$, receives the prepared system and, given an independently chosen random variable $y \in \mathcal{Y}$ as the choice of an observable, measures the system returning the output $b \in \mathcal{B}$ (see Fig. \ref{fig:pm-box}). Without further information on the inner workings of these devices, the best possible description is given by the probabilities $p(b,x,y)$ of the observed events. Without loss of generality, and given that the state preparation $x$ and measurement choices $y$ are under control of the experimenters, one typically considers the conditional distribution $p(b\vert x,y)$ represented by the set of behaviors  $\mathbf{p} = \{ p(b \vert x, y) \}_{b, x, y}$
    
    Classically, the state being prepared by the device $P$ is a random variable $a \in \mathcal{A}$, that can be understood as a message being sent to the $M$ device. Assuming that the preparation and measurement devices might have some pre-shared correlations, described by the random variable $\Lambda$ governed by a probability distribution $\pi(\lambda)$, the most general classical description of this experiment is given by
    \begin{equation}
        p(b \vert x, y) = \sum_a \int_\Lambda \pi(\lambda) p(a \vert x, \lambda) p(b \vert a, y, \lambda).
        \label{eq:classical-model}
    \end{equation}
    To obtain the decomposition above, a set of causal assumptions is considered. First the independence of the pre-shared correlations from the state preparation and measurement choices, that is, $p(x,y,\lambda)=p(x,y)p(\lambda)$; an assumption that in the context of Bell scenarios is referred as measurement indepedence or ``free-will'' \cite{hall2010local,wood2015lesson,chaves2015unifying}. The second causal assumption follows from the fact that, differently from a Bell scenario, in the prepare and measure case we deal with temporal correlations. More precisely, the measurement device lies in the future light cone of the preparation device. Thus, the message $a$ should only be causally dependent on the input $x$ and the shared correlations $\lambda$. The final causal assumption is given by the fact that even though the outcome $b$ can depend on the input $x$, such correlations are mediated by the state being prepared, that is, the correlations between $b$ and $x$ are screened-off once we condition on the values of $a$ and $\lambda$. Such causal assumptions can be faithfully and graphically represented using the intuitive directed acyclic graph (DAG) shown in Fig. \ref{fig:pm-dag}. 
    \begin{figure}
        \centering
        \subfigure[Box representation]{\label{fig:pm-box}\includegraphics[width=.42\columnwidth]{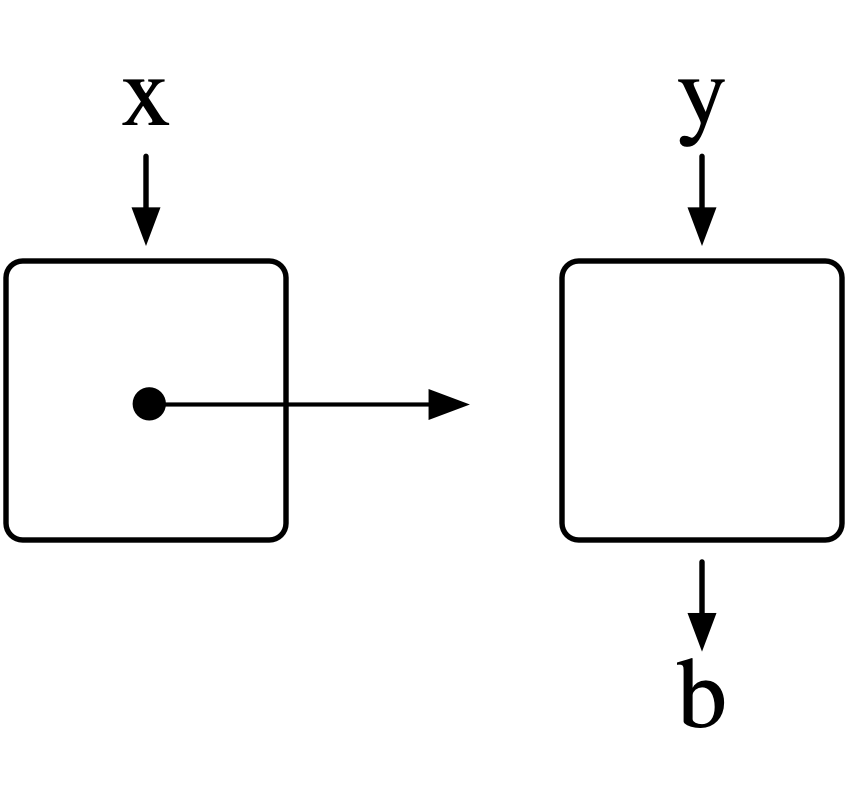}}\hspace{3.5em}
        \subfigure[DAG representation]{\label{fig:pm-dag}\includegraphics[width=.25\columnwidth]{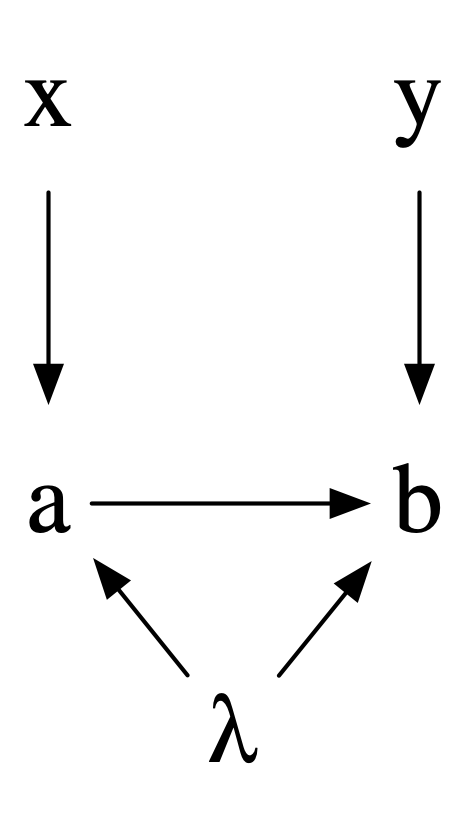}}
        \caption{Operational and causal representations of a prepare and measure scenario. On the left, the black-box shows a preparation, conditioned on input $x$, being sent to a measurement device with measurement choice $y$ and output $b$. On the right, the DAG turns explicit the unobservable variables $a$ and $\lambda$ in the causal structure.}
    \end{figure}

    In a quantum description, states from the set $\mathcal{S} = \{\rho_x \}_x$ are prepared and the possible observables are quantum measurements $\mathcal{M}_y$, where each $\mathcal{M}_y = \{ M_{b \vert y} \}_{b}$ is a collection of Positive Operator-Valued Measure (POVM) operators. The quantum experiment is then defined by $\mathcal{E} = \{ \mathcal{S}, \mathcal{M} \} $, where $ \mathcal{M} = \{ \mathcal{M}_y \}_y $. Employing Born's rule, any element $p(b \vert x, y)$ in the behavior $\bm{p}$ of experiment $\mathcal{E}$ is then given by
    \begin{equation}
        p(b \vert x, y) = \text{tr}\left( M_{b \mid y} \rho_x \right) .
    \end{equation}
    Notice that if $\abs{\mathcal{A}} \geq \abs{\mathcal{X}}$, a message $a$, even if classical, may encode the choice $x$ of preparation, and any behavior $\bm{p}$ is then reproducible with no more than classical communication. Only when some restriction is imposed on such communication, is that differences between the classical and quantum predictions can emerge \cite{gallego-pam-2010,bowles2014certifying,chaves2015device,van2017semi,tavakoli2020informationally,Poderini2020}. Typically, the bound $\abs{\mathcal{A}} < \abs{\mathcal{X}}$ is imposed on the dimension of the classical message. For instance, by preparing qubit states we can generate correlations not reproducible by classical bit messages \cite{gallego-pam-2010}. Whenever each and every $p(b \vert x, y)$ can be written as the classical model \eqref{eq:classical-model}, the quantum experiment $\mathcal{E}$ is classically reproducible, meaning the measurements statistics can be obtained by transmitting classical messages $a$ of the same dimension as the prepared quantum states in the set $\mathcal{S}$.


\section{Witnessing the classicality of preparations} 

    The first question one might ask in the PAM scenario is whether a quantum experiment $\mathcal{E}$, characterized by prepared states and measurements, can be classically simulated. Similarly to what happens in a Bell scenario, if the cardinalities of the sets of the variables $a$, $b$, $x$ and $y$ are fixed, the set of behaviors compatible with the classical description \eqref{eq:classical-model} is a polytope that can be characterized in terms of a finite number of linear (Bell-like) inequalities. If one has a complete description of such inequalities for a given scenario, then checking the compatibility of the quantum experiment data with a classical description can be easily certified if no inequalities are violated. The problem of this approach is twofold. First, determining all inequalities becomes intractable as we increase the cardinality of the variables. The second, more fundamental issue, stems from the fact that states that can only give rise to classical correlations if a given number of measurements are performed, can indeed have their non-classicality revealed if we increase the number of possible measurements \cite{Poderini2020}. In the following Theorem, we address this problem by deriving a sufficient condition to certify the classicality of a given set of prepared states valid for any number of projective measurements (PMs).
    \begin{thm}[Preparation classicality for all PMs]
        Let $\mathcal{S} = \{ \rho_x \}_{x=1}^{\abs{\mathcal{X}}}$ be a set of $d$-dimensional quantum preparations. All behaviors of $\mathcal{S}$ are classically reproducible for all projective measurements if there exist a finite collection $\mathcal{M}$ of projective, rank-1 measurements on $\mathbb{C}^d$ such that each $\{ O_x \}_{x=1}^{\abs{\mathcal{X}}}$ implicitly defined by
        \[\rho_x = \eta O_x + (1 - \eta) \frac{\bm{1}_d}{d} ,\]
        where $\eta$ is the radius of the largest sphere that can be inscribed into the convex hull of $\mathcal{M}$, admits a classical model for the measurements in $\mathcal{M}$.
        
        \label{thm:projective}
    \end{thm}
    \begin{proof}
        We begin by recalling that each $\Pi_{b \vert y} \in \mathcal{M}$ may be associated to a unit vector $\bm{v}_{b|y} \in \mathds{R}^{(d^{2}-1)}$ by means of
        \begin{equation}
            \Pi_{b|y} = \frac{1}{d}\left(\bm{1}_{d} + c_{d}\sum_{i=1}^{d^{2}-1}(v_{b|y})_{i} \sigma_{i}\right),
        \end{equation}
        where $c_{d} = \sqrt{d(d-1)/2}$ and $\left\{ \sigma_{i} \right\}_{i}^{(d^{2}-1)}$ is a set of $d \times d$ traceless operators (generators of the $SU(d)$ group) that, together with $\bm{1}_{d}$, form an orthonormal basis for the space of $d\times d$ linear operators (with respect to the Hilbert-Schmidt inner product).
        
        With that in mind, let $\mathcal{M} = \{ \Pi_{b \vert y} \}_{b, y}$ be a finite set of projective measurements, where each $\Pi_{b \vert y}$ is a rank-1 projector (we will later show that this leads to no loss of generality), and $\sum_b \Pi_{b \vert y} = \bm{1}_d, \,\forall y$. We will write $\eta$ for the radius of the largest sphere that can be inscribed into $\textbf{conv}(\mathcal{M}) = \textbf{conv}(\{ \bm{v}_{b|y} \}_{b,y} )$, and $\{ O_x \}_{x=1}^{\abs{\mathcal{X}}}$ for the set of operators implicitly defined by $\rho_x = \eta O_x + (1 - \eta) \frac{\bm{1}_d}{d}$.
        
        When model \eqref{eq:classical-model} exist for preparations $\mathcal{S}$ and measurements $\mathcal{M}$, it also does for any convex combination of the operators in $\mathcal{M}$. In particular, then, classicality follows for any measurements in the largest sphere that can be inscribed into $\textbf{conv}(\mathcal{M})$ (see Fig.~\ref{fig:measurements} for a representation in $d=2$). As any measurement operator on such a sphere is related to a valid rank-1 projector $\Pi_{b \mid u}$ through $\Pi_{b \vert u}^\eta = \eta \Pi_{b \mid u} + (1 - \eta) \mathbf{1}_d/d$, a simple calculation yields
        \begin{equation}
            \text{tr} \left( O_x \Pi_{b \vert u}^\eta \right) = \text{tr}\left( \rho_x \Pi_{b \vert u} \right), \quad \forall x, b, u .
            \label{eq:equal-traces}
        \end{equation}

        Eq. \eqref{eq:equal-traces} tells us the result of applying any projection $\Pi_{b \mid u}$ to a density operator $\rho_x$ is equivalent to applying a depolarized projection $\Pi_{b \vert u}^\eta$ to $O_x$ --- an inflated instance of $\rho_x$.
        
        With that in mind, suppose we probe each $O_x$ with every $\Pi_{b \vert y} \in \mathcal{M}$ and find the observed behavior has a classical model as given by eq. \eqref{eq:classical-model}. In that case, as already noted, the model would exist for all $\Pi_{b \vert u}^\eta$, as they all can be written as convex combinations of the $\Pi_{b \vert y}$. Thereon, eq. \eqref{eq:equal-traces} affirms the $\rho_x$ are classically simulable for all rank-1 projectors.
        
        To see this condition is also valid for arbitrary projections, first notice that in $\mathbb{C}^2$ all non-trivial projections are unit-rank. Even though this is not true for $d > 2$, it is true that any projective measurement can be seen as a rank-1 projective measurement with coarse graining. Hence, the result holds for all projective measurements.
    \end{proof}
    
    \begin{figure}
        \centering
        \includegraphics[width=.7\columnwidth]{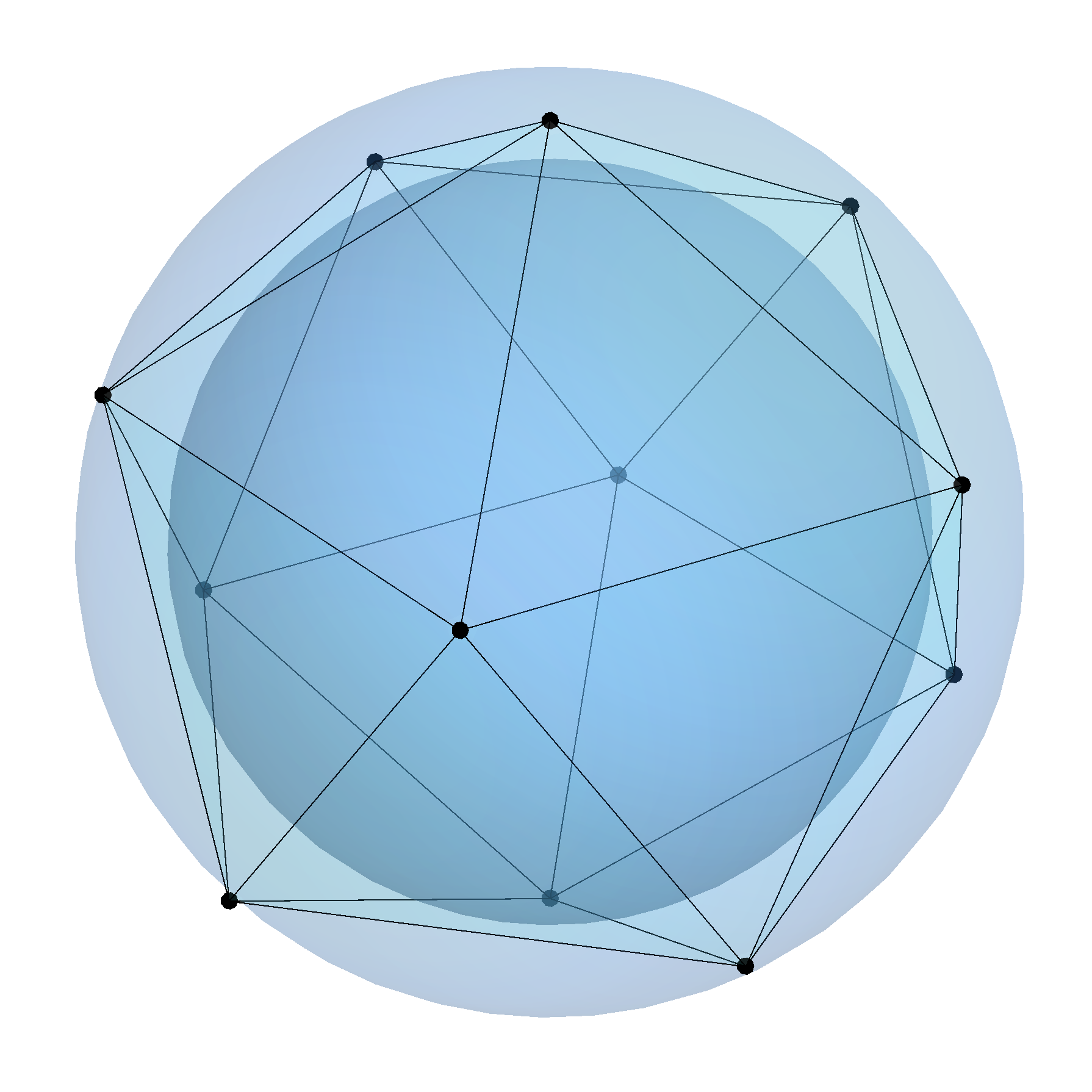}
        \caption{Representation of the method for $d = 2$. Each vertex on the Bloch sphere represents a measurement operator. To every $\Pi_{0 \vert y}$ we associate a corresponding antipodal $\Pi_{1 \vert y}$. Measurements inside the set enclosed by $\textbf{conv}(\mathcal{M})$ can be simulated by mixing these extremal ones. In particular, any measurement $\{ \Pi_{b \vert u}^\eta \}_b$ in a ball with radius $\eta$ inscribed in the polytope is simulable in such manner.}
        \label{fig:measurements}
    \end{figure}
    
    Strikingly, by probing the statistics of only finitely many measurements, Theorem \ref{thm:projective} provides us with a sufficient condition to certify that the set of prepared states $\mathcal{S} = \{ \rho_x \}_{x=1}^{\abs{\mathcal{X}}}$ can only exhibit classical correlations, even if infinitely many projective measurements were to be performed. This sufficient condition for classicality also becomes necessary when $\eta=1$, in which case we recover the brute-force approach of testing all possible measurements. Additionally, this classicality condition can be expressed as the following feasibility problem
    \begin{subequations}
		\begin{alignat}{2}
			&\text{given}    &\quad & \mathcal{S},\, \mathcal{M},\, \eta, \{ \lambda \} \\
	        &\text{find}   &	  & \pi(\lambda) \\
			&\text{s.t.}    &      & \rho_x = \eta O_x + \left( 1 - \eta \right) \frac{\bm{1}_d}{d}, \,\forall x \\
			&                  &      & \text{tr}(\Pi_{b \vert y} O_x) = \sum_{a, \lambda} \pi(\lambda) p(a \vert x, \lambda) p(b \vert a, y, \lambda), \,\forall\, b, x, y \label{eq:program-model1}\\
			&				   &	  & \pi(\lambda) \geq 0, \\
			&				   &	  & \sum_\lambda \pi(\lambda) = 1,
		\end{alignat}
		\label{eq:factibility}
	\end{subequations}
    which is a linear program and thus can be solved efficiently. In particular, condition \eqref{eq:program-model1} enforces the existence of the classical model \eqref{eq:classical-model}, and Fine's theorem was used to cast the integral as a finite sum \cite{fine-hidden-1982}.
    
    Dichotomic projective measurements are the extremal two-effect POVMs, so Theorem \ref{thm:projective} is also a condition for classicality under all such measurements. To extend our result to generalized measurements, we observe that any POVM collection $\mathcal{M} \subseteq \mathcal{P}(d,n)$, where $\mathcal{P}(d,n)$ is the set of generalized measurements with $n$ effects acting on $d$-dimensional preparations, can be simulated by projective measurements and classical processing after a certain amount $t$ of depolarization on its effects \cite{oszmaniec-simulating-2017}, which leads to the following extension of Theorem \ref{thm:projective}:
    \begin{thm}[Preparation classicality for all POVMs]
        Let $\Phi_{t}(\mathcal{M})$, where $\Phi_t(\cdot) = t(\cdot) + (1-t) \frac{\text{tr}(\cdot)}{d}\bm{1}_d$ is a depolarizing channel acting on $\mathcal{M}$'s effects, be projective-simulable for any $\mathcal{M} \subseteq \mathcal{P}(d,n)$. If preparation $\rho^\prime = \frac{1}{t} \left( \rho - \frac{1-t}{d} \bm{1}_d \right)$ has a classical model for all projective measurements, then $\rho$ is classically reproducible for all POVMs.
        \label{thm:povms}
    \end{thm}
    \begin{proof}
        Let $\rho^\prime$ be such that probabilities $\text{tr}(\Pi_{b \vert y} \rho^\prime)$ admit description \eqref{eq:classical-model} for all projections $\Pi_{b \vert y}$, and $\mathcal{M} = \{ M_{b \mid y} \}_{b, y}$ be a collection of POVMs. As the statistics generated by $\Phi_{t}(\mathcal{M})$ for any $\mathcal{M} \subseteq \mathcal{P}(d,n)$ can be reproduced by projective measurements and classical processing, it must be that the behavior $\{ \text{tr}[\Phi_{t}(M_{b \mid y}) \rho^\prime] \}_{b, y}$ is also classically reproducible. But $\Phi_t(\cdot)$ is self-dual, hence $\text{tr}[\Phi_{t}(M_{b \mid y}) \rho^\prime] = \text{tr}[M_{b \mid y} \Phi_{t}(\rho^\prime)] = \text{tr}(M_{b \mid y} \rho), \forall b, y$. Therefore, $\rho$ has a classical model for all POVMs.
    \end{proof}
    
    Accordingly, showing that preparations $\rho_x$ have a classical description for all POVMs is equivalent to proving that preparations $\rho_x^\prime = \frac{1}{t} \left( \rho_x - \frac{1-t}{d} \bm{1}_d \right)$ are classically reproducible for all projective measurements, which can be done by providing $t$ and the relevant additional restrictions to program \eqref{eq:factibility}. 
    For $d=2$, the POVM $A_{\text{tetra}} = \{ \frac{1}{4}\left( \bm{1} + \bm{v}_i \cdot \sigma \right) \}_i$, where $\bm{v}_i$ are the vertices of a regular tetrahedron, have the highest noise robustness, thus any $t$ such that $\Phi_t(A_{\text{tetra}})$ is projective-simulable also makes so any $A \in \mathcal{P}(2, n)$. For qubits, $t = \sqrt{2/3} - \epsilon$ suffices. On higher $d$, $t=1/d$ is a lower bound on $t$ that can be tightened through semidefinite programming \cite{oszmaniec-simulating-2017}.
    
\begin{figure*}[t!]
        \centering
        \subfigure[Icosahedron]{\label{fig:hm-icos}\includegraphics[height=.31\linewidth]{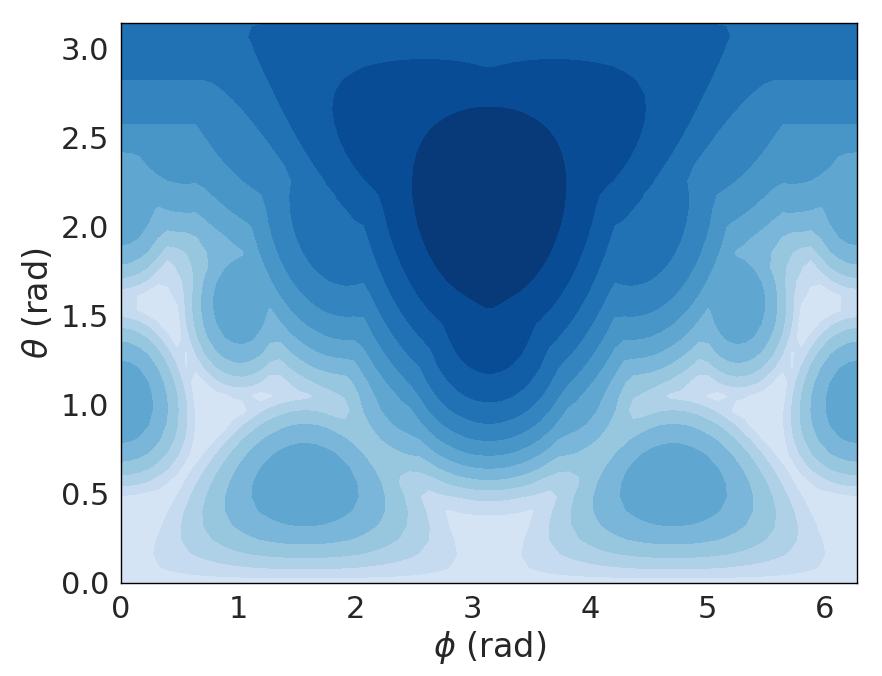}}
        \subfigure[Rhombicuboctahedron]{\label{fig:hm-romb}\includegraphics[height=.31\linewidth]{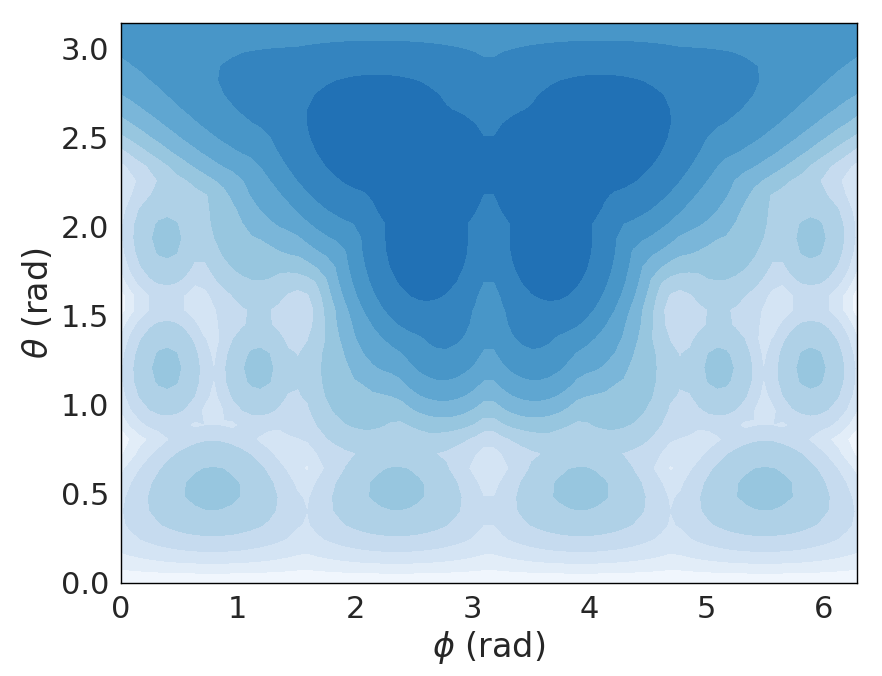}}
        \subfigure{\includegraphics[height=.31\linewidth]{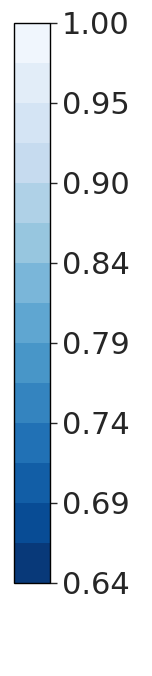}}
        \caption{Application of program \eqref{eq:maximization} to $\mathcal{S}(\theta, \phi) = \{ \rho_{\bm{x}}, \rho_{\bm{z}}, \rho_{\bm{r}(\theta, \phi)} \}$. Levels are the maximum visibility $\alpha$ such that preparation set $\alpha \mathcal{S}(\theta, \phi)$ has a classical model. \subref{fig:hm-icos} For the $\abs{\mathcal{Y}} = 6$ icosahedron measurements, $\eta \approx 0.79$, and program \eqref{eq:maximization} can be directly applied. \subref{fig:hm-romb} A rhombicuboctahedron corresponds to $\abs{\mathcal{Y}} = 12$ projective measurements and $\eta \approx 0.86$, but as the number of deterministic strategies scales exponentially, the computation is only possible by iteratively optimizing over subsets of deterministic strategies.}
    \end{figure*}

 \subsection{Computational analysis}
 The computational bottleneck of program \eqref{eq:factibility} lies at the fact that decomposing the deterministic strategies $\lambda$ into their extremal points results in an exponentially large set, with $N_\lambda \propto \abs{\mathcal{B}}^{\abs{\mathcal{A}} \abs{\mathcal{Y}}}$ extremal points. Consequently, regardless of the efficiency of linear programming algorithms, the size of program \eqref{eq:factibility} scales exponentially in the number of measurements. 
    In principle, this precludes us from using too large sets of measurements, and consequently obtaining larger  values of $\eta$, but we can circumvent this issue by adapting the procedure outlined in \cite{fillettaz-algorithmic-2018}.
    
    Our strategy will be to avoid working on all $N_\lambda$ extremal points at once, and instead iteratively exploring the whole deterministic strategy space. Notice that the factilibity program \eqref{eq:factibility} may be equivalently written as a maximization program,
    \begin{subequations}
    	\begin{alignat}{2}
    		&\text{given}    &\quad & \mathcal{S},\, \mathcal{M},\, \eta, \{ \lambda \} \\
            &\underset{\pi(\lambda)}{\text{max.}}   &	  & \alpha \\
    		&\text{s.t.}    &      & \alpha \rho_x + (1 - \alpha)\frac{\bm{1}_d}{d} = \eta O_x + \left( 1 - \eta \right) \frac{\bm{1}_d}{d}, \,\forall x \\
    		&                  &      & \text{tr}(\Pi_{b \vert y} O_x) = \sum_{a, \lambda} \pi(\lambda) p(a \vert x, \lambda) p(b \vert a, y, \lambda), \,\forall\, b, x, y \\
    		&				   &	  & 0 \leq \alpha \leq 1 \\
    		&				   &	  & \pi(\lambda) \geq 0 \\
    		&				   &	  & \sum_\lambda \pi(\lambda) = 1
    	\end{alignat}
    	\label{eq:maximization}
    \end{subequations}
    Allowing $\alpha = 0$ guarantees a solution will always exist, and obtaining $\alpha = 1$ amounts to \eqref{eq:factibility} being factible. One may interpret the program above as a search for the optimal weights $\pi(\lambda)$ such that a convex combination of extremal points of the local polytope describe the behavior of our system. Carathéodory's Theorem \cite{rockafellar-convex-1970} states that at most $d+1$ extremal points are necessary to optimally describe any point of a $d$-dimensional convex set, hence most of the $\pi(\lambda)$ found will be zero. We cannot know, beforehand, which points make for an optimal description, so we take $N_\lambda^\prime \gg d+1$ --- but much smaller than $N_\lambda$ --- points and optimize over them. To set up the next iteration, all $\pi(\lambda) = 0$ in this result can be discarded and replaced by previously unexplored deterministic strategies, and we run program \eqref{eq:maximization} again. As, at each round, we are keeping all optimal $\lambda$ from the previous, the optimal value $\alpha$ will be non-decreasing between iterations. Furthermore, for so large $\abs{\mathcal{Y}}$ that it would be prohibitive to enumerate and keep track of all previously visited strategies, we observe that simply randomly sampling $\lambda$'s on each run of \eqref{eq:maximization} makes our procedure rapidly converge to $\alpha^*$, which assumes a constant value for all subsequent iterations and is interpreted as a lower bound on the maximum visibility imposed on the preparations such that their behavior is classical.

    To illustrate the application of Theorem \ref{thm:projective} and the procedure just described, we consider the preparation set $\mathcal{S}(\theta, \phi) = \{ \rho_{\bm{x}}, \rho_{\bm{z}}, \rho_{\bm{r}(\theta, \phi)} \}$ (that is, $\vert \mathcal{X}\vert=3$), where $\rho_{\bm{v}}$ denotes a qubit state with Bloch vector $\bm{v}$, and the unit vector $\bm{r}(\theta, \phi)$ is given by its spherical polar and azimuthal angles, respectively. Arranging our measurements operators as those associated to the vertices of an icosahedron ($\abs{\mathcal{Y}} = 6$ with $\eta \approx 0.79$), we get a fairly small problem that can be solved either directly or iteratively (Fig.~\ref{fig:hm-icos}). With twice the amount of measurements arranged as the vertices of a rhombicuboctahedron ($\eta \approx 0.86$), it is only possible to compute fig. \ref{fig:hm-romb} using the iterative procedure. The advantage, however, is remarkable, with more measurements resulting in increased visibilities.


    \subsection{Non-classicality activation}
    
        In a Bell scenario, an entangled state that can only lead to local correlations can have its non-locality activated in at least two manners: with a single copy, by proceeding with local filtering (a phenomenon usually referred as hidden nonlocality) \cite{popescu-filtering-1995, hirsch-hidden-2013,gallego2014nonlocality}, or by using many copies of the quantum state \cite{navascues-activation-2011, palazuelos-superactivation-2012,cavalcanti2011quantum}. In turn, in a PAM scenario a new possibility is open. Consider a set of $n+1$ prepared states such that any subset of  $n$ of them can only generate correlations describable by the classical model \eqref{eq:classical-model}. However, if the $n+1$ quantum states can lead to non-classical correlations then we can say the non-classicality of the other $n$ states is activated by the preparation of this extra $(n+1)$-th state. For our purposes it will suffice to consider the case where $n=3$, that is, we have four possible preparations ($\vert \mathcal{X} \vert=4$) such that any combination of only three of them will always lead to classical correlations. All other variables are dichotomic, that is, $\abs{\mathcal{B}} = \abs{\mathcal{A}} = \abs{\mathcal{Y}} = 2$.
    
        In this specific PAM scenario, there are only two classes of inequalities, the violation of which certify non-classicality. If we are able to find a set of four states, any three of which only generate classical correlations according to our criteria, but nonetheless violate one of these inequalities when taken together, we would thus have proven the activation of non-classicality. Towards that end, we consider the inequality given by \cite{pawlowski-qkd-2011}
        \begin{equation}
            \begin{split}
                S = E_{11} - E_{12} - E_{21} + E_{22} + \\ - E_{31} - E_{32} + E_{41} + E_{42} \leq 4,
            \end{split}
            \label{eq:inequality}
        \end{equation}
        where $E_{xy} = p(0 \vert x, y) - p(1 \vert x, y)$ is the expectation value of observable $y$ applied to preparation $x$. When preparations are quantum states and the measurements are extremal, $E_{xy} = \text{tr}(M_{0 \vert y} \rho_x) - (M_{1 \vert y} \rho_x) = \bm{r}_x \cdot \bm{q}_y$, where $\bm{r}_x$, $\bm{q}_y$ are the Bloch vectors that parametrize preparation $x$ and measurement $y$, respectively.
    
        We define $\mathcal{S}(\alpha, \theta) = \{ \rho_{\bm{r_1}}, \rho_{\bm{r_2}}, \rho_{\bm{r_3}}, \rho_{\bm{r_4}} \}$ as the preparation set according to Fig. \ref{fig:preparations}, where $\alpha$ is a shrinking factor from the surface of the Bloch sphere. In turn, we choose projective measurements parameterized by the vectors $\bm{q}_1 = - \bm{x}$ and  $\bm{q}_2 = \bm{z}$, for which we obtain that $S = 4\sqrt{2} \alpha \sin \theta$. This shows that, for a large span of $\alpha$ and $\theta$, this family of preparations violate the bound $S\leq 4$, thus exhibiting non-classicality.
        
        \begin{figure}
            \centering
            \includegraphics[width=.8\columnwidth]{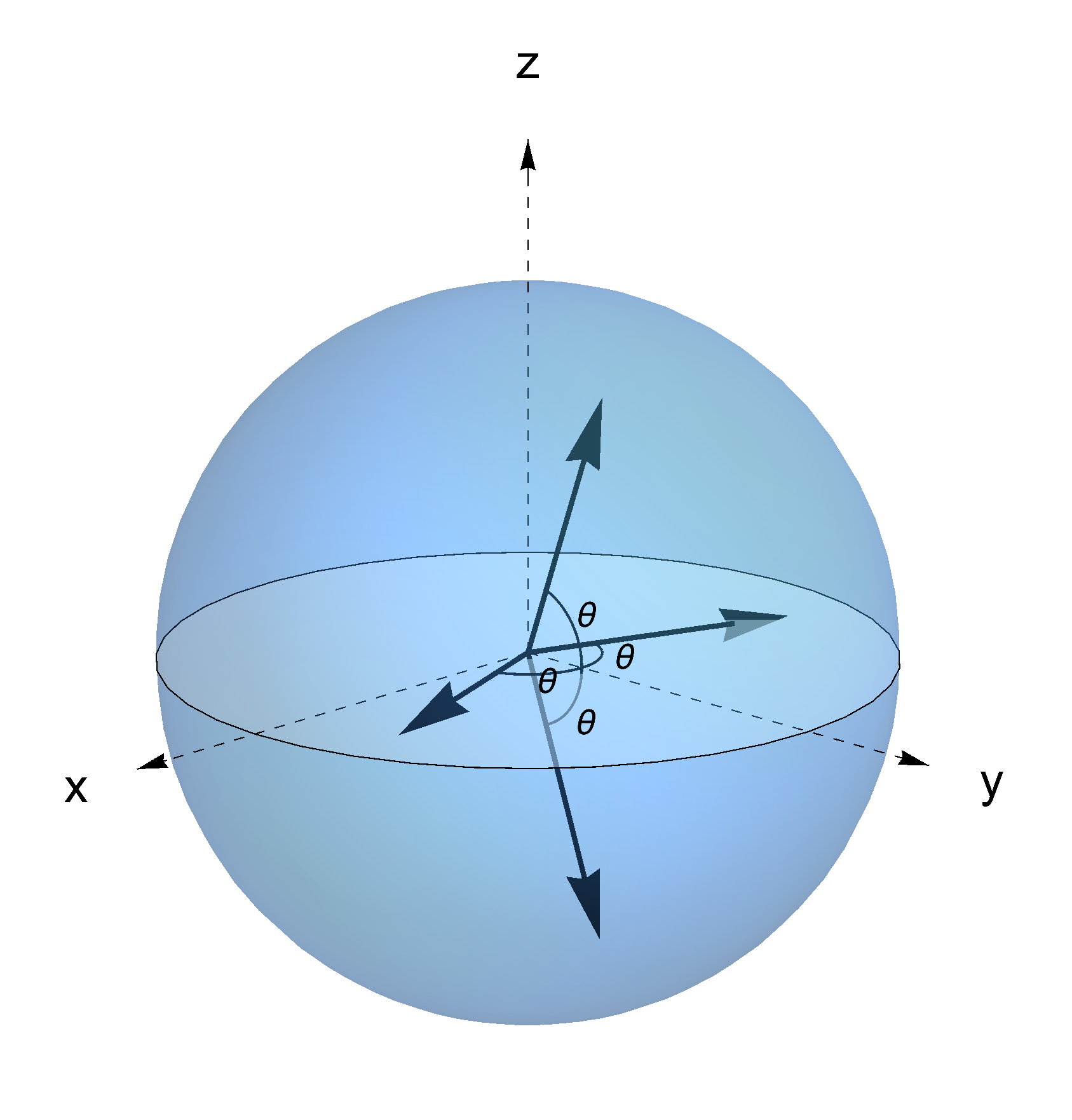}
            \caption{Preparations $\mathcal{S}(\alpha, \theta) = \{ \rho_{\bm{r_1}}, \rho_{\bm{r_2}}, \rho_{\bm{r_3}}, \rho_{\bm{r_4}} \}$ for $\alpha=0.8$. At $\theta=0$ all preparations are at $\alpha \bm{y}$. For $\theta = \pi/2$, $\mathcal{S} =  \{ -\alpha\bm{x}, \alpha\bm{x}, -\alpha\bm{z}, \alpha\bm{z} \}$, presenting the largest violation of inequality \eqref{eq:inequality}.}
            \label{fig:preparations}
        \end{figure}
        
        To show this non-classicality is a genuine non-classicality activation for dichotomic measurements, we employ our general method to show the behaviors of any subset of three elements of $\mathcal{S}$ are classically reproducible. This was done by applying program \eqref{eq:maximization} with measurements arranged as a rhombicuboctahedron, and corresponding $\eta \approx 0.86$. As a result, for each $\theta$, the optimal value $\alpha^*$ stands for the maximum purity of the preparations such that there is a classical model for all triadic subsets of $\mathcal{S}$, and every $\alpha < \alpha^*$ represents classical preparations (fig. \ref{fig:activation}). As shown in the shaded region, we will have non-classicality activation for any value $\alpha < \alpha^*$ such that inequality $S$ is violated.

    \subsection{Quantum advantage activation in RACs}
    \label{sec:racs}
            
        The non-classicality results shown here also have implications for the application of the PAM scenario in random access codes (RACs) \cite{li2012semi, pawlowski-qkd-2011}. An $n \mapsto i$ RAC can be understood as a communication task where, at each run, Alice receives a string $(m_1, \dots, m_n)$, with each $m_i \in \{ 1, \ldots, d \}$, which she then encodes in $i < n$ either classical or quantum $d$-level systems that will be sent to Bob. Upon receiving this limited communication from Alice and an uniformly sampled query $y \in \{1, \ldots, n \}$, Bob is asked to determine the letter $m_y$ from her string. Their joint task is to maximize the average probability of success of his guess, $b$, of the input bits of Alice, over all encoding/decoding strategies. The figure of merit is given by
        \begin{equation}
            p_{suc} = \frac{1}{n d^n} \sum_{m_1, \ldots, m_n, y} p(b = m_y \mid m_1,\ldots, m_n, y ),
            \label{eq:psuc}
        \end{equation}
        where the factor $\frac{1}{nd^{n}}$ comes from the assumption that the input dits as well as the variable $y$ are uniformly distributed. 
        
     \begin{figure}
            \centering
            \includegraphics[width=\columnwidth]{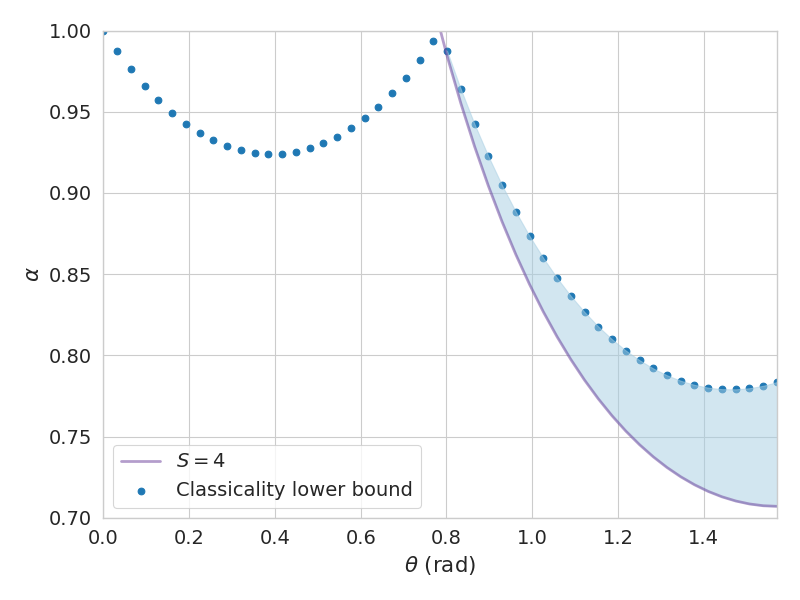}
            \caption{Non-classicality activation for the preparations in fig. \ref{fig:preparations} and measurements arranged in a rhombicuboctahedron with corresponding $\eta \approx 0.86$. As $S \propto \alpha$, every preparation set above the $S=4$ curve is non-classical. On the other hand, the blue scatter shows the maximum visibility such that any triad of states in the preparation set are classical. The shaded region thence stands for preparations that exhibit non-classicality activation by increasing the number of preparations.}
            \label{fig:activation}
        \end{figure}

        A PAM scenario with $\vert \mathcal{X} \vert = 4 $ may be mapped to a $2 \mapsto 1$ binary RAC by, for instance, mapping $x = 1 \mapsto (m_1 = 0, m_2 = 0)$, $x = 2 \mapsto (m_1 = 0, m_2 = 1)$, $x = 3 \mapsto (m_1 = 1, m_2 = 0)$, $x = 4 \mapsto (m_1 = 1, m_2 = 1)$. In this case, as shown in \cite{pawlowski-qkd-2011}, the probability of success \eqref{eq:psuc} can be directly linked with inequality \eqref{eq:inequality} in a way that $p_{suc} = \frac{S+8}{16}$. Indeed, the classical PAM bound $S=4$ corresponds to the optimal probability of success of $p_{suc}=3/4$, achievable when Alice sends a classical bit. On the other hand, for any quantum violation of the inequality \eqref{eq:inequality}, we do have a corresponding quantum advantage in a quantum RAC. In light of our non-classicality activation result, this means that we have identified sets of four qubits states being prepared by Alice, such that any three of them are classical but altogether offer an advantage in a relevant quantum communication protocol.


\section{Witnessing the classicality of measurements}

    As much as quantum states may be seen as resources in diverse correlation scenarios, one may also look the other way around and inquire whether a given set of quantum measurements are useful for unveiling non-classicality. For the former, we have shown that Theorem \ref{thm:projective} allows to certify  that a preparation set only generates classical statistics regardless of which and how many projective measurements are applied. It may, however, be modified to treat measurements as resources, and certify a set of projective measurements can only given rise to classical statistics for all possible preparations.

    \begin{thm}[PMs classicality for all preparations]
        Let $\mathcal{M} = \{ \Pi_{b \mid y } \}_{b, y}$ be a collection of projectors on $\mathbb{C}^d$ such that $\sum_{b} \Pi_{b \mid y } = \bm{1}_d, \forall y $, and $\mathcal{S} = \{ \rho_x \}_x$ be a finite collection of pure $d$-dimensional preparations. Furthermore, define operators $O_{b \mid y}$ through
            \[ \Pi_{b \mid y} = \eta O_{b \mid y} + (1 - \eta) \frac{\bm{1}_d}{d} , \]
        where $\eta$ is the radius of the largest sphere that can be inscribed into the convex hull of $\mathcal{S}$. If there exists a probability distribution $\pi(\lambda)$ such that that all $p(b \vert x, y) = \text{tr}(O_{b\mid y} \rho_x)$ can be written as eq. \eqref{eq:classical-model}, then measurements $\mathcal{M}$ can never manifest non-classical statistics for any preparation set.
        \label{thm:pms-classicality1}
    \end{thm}
    \begin{proof}
        Noticing that (i) whenever model \eqref{eq:classical-model} may be used to describe the behavior of $\mathcal{S}$ under a given set of measurements, then any preparation set in $\textbf{conv}(\mathcal{S})$ is also classically reproducible when probed by those same measurements, (ii) that any density operator on the shrunken Bloch sphere of radius $\eta$ can be written as $\rho_{\bm{u}}^\eta = \eta \rho_{\bm{u}} + (1 - \eta) \bm{1}_d/d$ and is in $\textbf{conv}(\mathcal{S})$, and (iii) that, for rank-1 projectors, $\text{tr}(O_{b\mid y} \rho_{\bm{u}}^\eta) = \text{tr}(\Pi_{b\mid y} \rho_{\bm{u}})$, the argument follows analogously to Theorem \ref{thm:projective}'s: if by probing operators $O_{b \mid y}$ with a finite set $\mathcal{S}$ of preparations we find out a classical PAM model \eqref{eq:classical-model} exists, then by (i) and (ii) the model exists for all $O_{b \mid y}$ and any $\rho_{\bm{u}}^\eta$, which by (iii) means it exists for $\mathcal{M}$ and all pure preparations. Invoking the same reasoning as before to extend the result for projections of any rank and observing that, if $\mathcal{M}$ is classically reproducible for all pure states, it also is for convex combinations of pure states, then $\mathcal{M}$ is classically reproducible for all quantum states.
    \end{proof}
    
    Theorem \ref{thm:pms-classicality1} shows that by probing the behavior of a given set $\mathcal{M}$ of measurements on a finite set of preparations, we may certify $\mathcal{M}$ is PAM-classical, meaning it never generates non-classical statistics in the prepare and measure scenario, regardless what preparation set it acts upon. The choice of the probe preparations $\rho_x$ should be done as to maximize $\eta$, similarly to how we have done when choosing the probe measurements in Fig. \ref{fig:measurements}. When those are given, we arrive at the following linear factibility program.
    \begin{subequations}
		\begin{alignat}{2}
			&\text{given}    &\quad & \mathcal{S},\, \mathcal{M},\, \eta \\
	        &\text{find}   &	  & \pi(\lambda) \\
			&\text{s.t.}    &      & \Pi_{b \mid y} = \eta O_{b \mid y} + (1 - \eta) \frac{\bm{1}_d}{d}, \,\forall b, y \\
			&                  &      & \text{tr}(O_{b \mid y} \rho_x) = \sum_{a, \lambda} \pi(\lambda) p(a \vert x, \lambda) p(b \vert a, y, \lambda), \,\forall\, b, x, y \label{eq:program-model}\\
			&				   &	  & \pi(\lambda) \geq 0, \\
			&				   &	  & \sum_\lambda \pi(\lambda) = 1,
		\end{alignat}
		\label{eq:classical-measurements-maximization}
	\end{subequations}
    which can also be cast as a maximization program (cf. \ref{eq:maximization}) 
    to which the procedure of iteratively exploring the $\lambda$-space may likewise be applied. Withal, analogous considerations to the ones in theorem \ref{thm:povms} extend this certification of measurement classicality to generalized measurements.
    \begin{thm}[POVMs classicality for all preparations]
        If a generalized measurements set $\mathcal{M}$ is projective-simulable for a given amount $t$ of depolarization, and measurements $\phi_t(\mathcal{M})$ are certifiably classical for all operators in an inflated, generalized Bloch ball of radius $1/t$, then $\mathcal{M}$ is classical for all quantum preparations. 
        \label{thm:povms-classicality}
    \end{thm}
    \begin{proof}
        Given a set $\mathcal{S}^t = \{ \rho_x^t \}_x = \{ \frac{1}{t} \left[ \rho_x - \frac{1-t}{d} \bm{1}_d \right] \}_x$ of probe operators, where the $\rho_x$ are pure quantum states, theorem \ref{thm:pms-classicality1} may be employed to determine whether the projective-simulable $\phi_t(\mathcal{M})$ are classical in relation to all operators $\rho^t$ on the $1/t$-radius Bloch sphere where the operators in $\mathcal{S}^t$ lie. Observing that $\text{tr}[ \phi_t( M_{b \mid y}) \rho^t ] = \text{tr}[ M_{b \mid y} \phi_t(\rho^t) ] = \text{tr}( M_{b \mid y} \rho ), \,\forall \rho^t$, this is equivalent to determining the classicality of generalized measurements $\mathcal{M}$ in relation to any $\rho$, which are all possible quantum preparations.
    \end{proof}
    
    We finally remark that, while theorem \ref{thm:pms-classicality1} was stated in regard to \emph{pure} probes, in principle there is no issue in probing with the $\rho^t \in  \mathcal{S}^t$, as all quantum states are more mixed than those. 

    
    \subsection{Incompatible classical measurements}
    
        Quantum formalism is inlaid with the existence of quantities that may not simultaneously be known with arbitrary precision, which is one of many ways it defies our intuition. Any set of measurements with this property is called an incompatible measurement set, with compatibility being the opposite concept. Measurement incompatibility is known to be a necessary but insufficient condition for non-classicality manifestations in Bell non-locality \cite{bene-incompatibility-2018,hirsch-incompatibility-2018}, and it was also shown to be necessary \emph{and} sufficient for quantum advantage in EPR steering scenarios \cite{quintino-incompatibility-2014} and in $2 \mapsto 1$ binary RACs (see proposition 2 in \cite{carmeli-racs-2020}). Section \ref{sec:racs}'s observations shows us incompatibility is thereby necessary and sufficient for non-classicality in the aforementioned $\abs{\mathcal{X}} = 4$, $\abs{\mathcal{B}} = \abs{\mathcal{A}} = \abs{\mathcal{Y}} = 2$ PAM scenario. That sufficiency does not hold in general will be shown through an application of our measurements' classicality certification method.
        
        Measurement compatibility may be understood as joint measurability. Let $\mathcal{M} =  \{ M_{b \mid y} \}_{b, y}$ be any set of measurements, with each having the same number of outcomes for convenience. Whenever a $J_{\bm{\ell}}$, with $\bm{\ell} = \ell_1 \ell_2 \ldots \ell_\abs{\mathcal{Y}}$, and each $\ell_i \in \{ 1, \ldots, \abs{\mathcal{B}} \}$, is a positive semidefinite operator such that $\sum_{\bm{\ell}} J_{\bm{\ell}} = \bm{1}$ and $\sum_{\bm{\ell}} J_{\bm{\ell}} \delta_{\ell_x, a} = M_{a \mid x}$, then $\mathcal{M}$ is said to be jointly measurable, in the sense that the so-called parent measurement $J_{\bm{\ell}}$ is a single, well defined measurement from which every $M_{a \mid x}$ may be recovered. Whenever such a $J_{\bm{\ell}}$ does not exist, $\mathcal{M}$ is incompatible, or not jointly measurable.
        
        Another invaluable concept is that of incompatibility robustness $\chi^*_\mathcal{M}$ \cite{designolle-incompatibility-2019}, measuring how compatible measurements $\mathcal{M}$ are through
        \begin{equation}
            \chi^*_\mathcal{M} = \sup_{\substack{\chi \,\in\, [0, 1]\\ \{N_{b \mid y}\} \,\in\, \textbf{N}( \mathcal{M} )}} \{ \chi \mid \chi \{ M_{b \mid y} \} + (1 - \chi) \{ N_{b \mid y} \} \in \textbf{JM} \} .
            \label{eq:incompatibility-robustness}
        \end{equation}
        In this definition, $\textbf{JM}$ is the set of jointly measurable measurements, and $\textbf{N}$ is a noise model that, possibly depending on the $M_{b \mid y}$, determines the noise set, which must contain at least one jointly measurable set of measurements. Given $\textbf{N}$, the lower the $\chi^*_\mathcal{M}$, the more incompatible the measurements are, with $\chi^*_\mathcal{M} = 1$ if and only if our measurement set is jointly measurable. For closed noise sets, eq. \eqref{eq:incompatibility-robustness} turns into a maximization problem that may be written as a semidefinite program (see Appendix E in \cite{designolle-incompatibility-2019}). This is true, in particular, for the identity noise, $N_{b|y} = \bm{1}/\abs{\mathcal{B}}$, which is a common choice when unbiased noise is of interest.
        
        Determining $\chi^*_\mathcal{M}$ thence amounts to choosing a noise model with the required properties and optimizing \eqref{eq:incompatibility-robustness} via semidefinite programming. Defining $\mathcal{M}$ as the mirror-symmetric measurements shown in fig. \ref{fig:mirror-symmetric-measurements} and choosing a random noise map, $\chi^*_\mathcal{M}$ behaves as shown in Fig. \ref{fig:incompatibility-vs-classicality}. Any value of $\chi$ above the incompatibility robustness curve represents incompatible measurements. Applying program \eqref{eq:classical-measurements-maximization} as a maximization problem to these same measurements, we obtain lower bounds for their classicality, meaning that everything under the measurement classicality curve stands for measurements which are not able to generate non-classical statistics, irrespective of the preparation set we choose. As, in the shaded region, this is above the incompatibility curve, we conclude there are incompatible measurement sets for which no preparations can exhibit advantage over communicating through two-dimensional classical systems. Put in other words, incompatibility is not in general sufficient for non-classicality in the prepare and measure scenario.
        
        \begin{figure}
            \centering
            \includegraphics[width=.8\columnwidth]{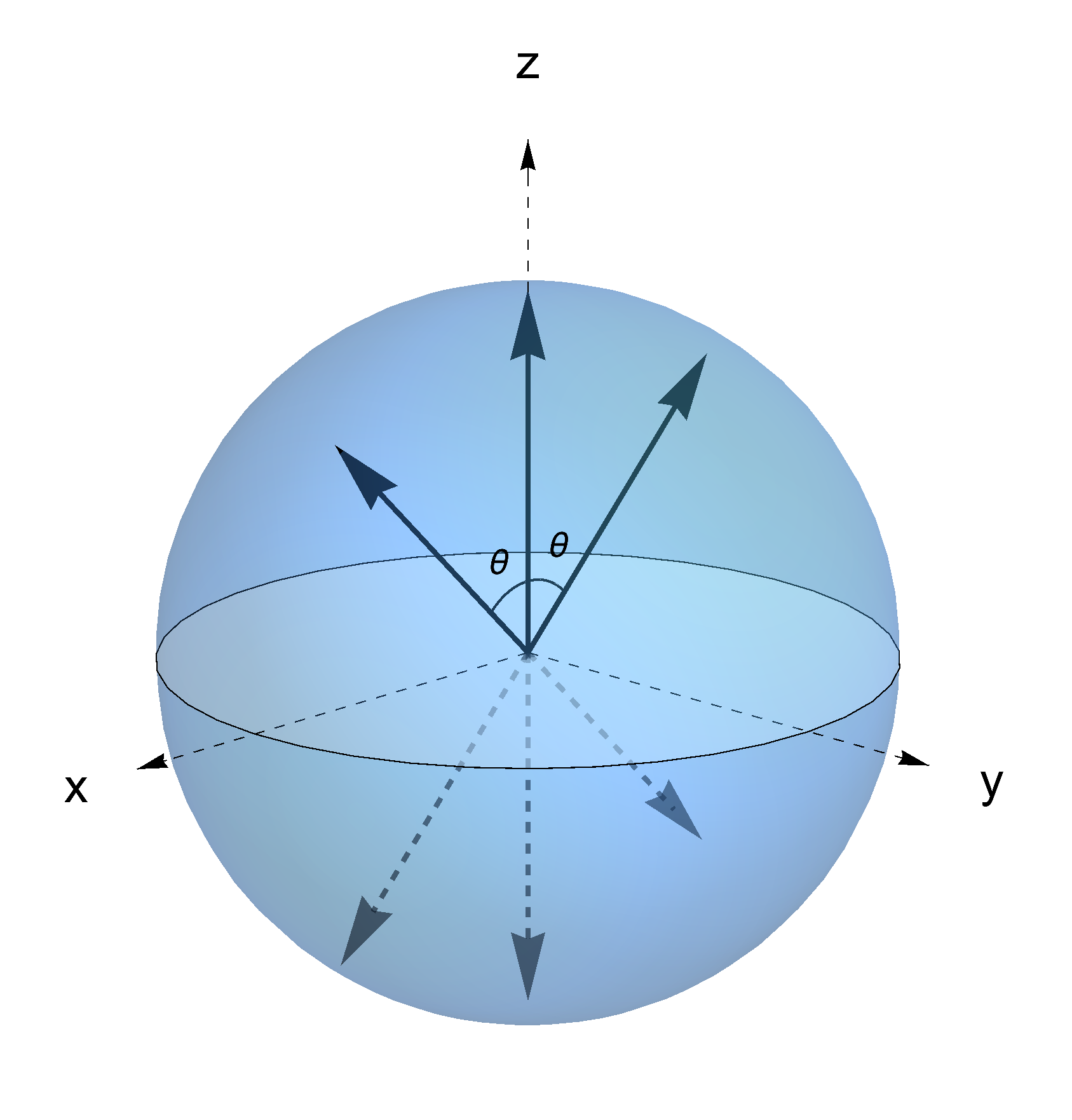}
            \caption{Mirror-symmetric measurements used to show the existence of incompatible measurements that do not exhibit non-classical statistics for any preparation set.}
            \label{fig:mirror-symmetric-measurements}
        \end{figure}
        
        \begin{figure}
            \centering
            \includegraphics[width=\columnwidth]{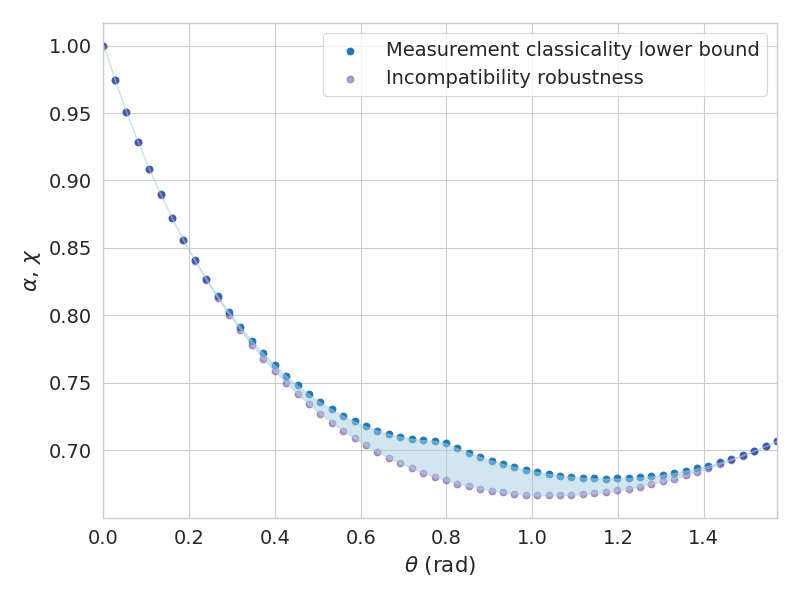}
            \caption{Incompatibility is insufficient for non-classicality in the prepare and measure scenario. For each $\theta$ (see fig. \ref{fig:mirror-symmetric-measurements}), any $\chi$ above the incompatibility robustness curve stands for an incompatible measurement set, and any $\alpha$ below the measurement classicality lower bound represents measurements that certifiedly do not generate non-classical statistics, regardless what preparations they act upon. Thenceforth, the shaded region contains incompatible albeit classical measurements.}
            \label{fig:incompatibility-vs-classicality}
        \end{figure}


\section{Conclusion}
    
    The ability to certify classicality is essential to known applications of the PAM scenario, which range from communication in quantum networks and self-testing of quantum channels to randomness certification and beyond. Although Bell-like inequalities for some modest settings in the prepare and measure scenario were already known, and activation phenomena were found for quantum preparations under two projective measurements, a method to certify classicality for any number of measurements --- and consequently a proof of genuine non-classicality activation for dichotomic measurements, --- remained elusive.
    
    We contributed to this problem devising a sufficient criterion to certify classicality for both projective and generalized measurements by only probing the statistics of finitely many measurements, then showing activation indeed happens for all dichotomic measurements in a large set of quantum preparations --- a result intimately connected to quantum advantages in random access codes. Using an optimization strategy inspired by \cite{fillettaz-algorithmic-2018}, we were able to increase the number of measurements we probe --- an indispensable ingredient to our method --- to otherwise intractable values. This was essential to the applications we have shown, and is straightforwardly adaptable to other scenarios such as Bell nonlocality and EPR steering \cite{cavalcanti-method-2016,hirsch-method-2016}.

   In turn, studying measurements as resources has been an active research topic \cite{Oszmaniec_2019, Buscemi-2020}, and following this trend we adapted our method to certify a given set of measurements is never able to generate non-classical behaviors, irrespective of what preparation set they are applied to. We showed the value of this tool by proving there are incompatible measurements which can only lead to classical correlations, which means measurement incompatibility is not sufficient for non-classicality in the PAM scenario.
    
    Further interesting possibilities are that of activation phenomena under generalized measurements with more than two effects, and activation of measurements non-classicality (similarly to what we have proven for the non-classicality activation of states). We were not able to show they happen for our targeted scenarios, but investigating larger settings is an interesting next step to which our methods are readily applicable. Additionally, figs. \ref{fig:activation} and \ref{fig:incompatibility-vs-classicality} shows our results are both noise and preparation error resistant for a large span of states, and experimental implementation would require no entangled states. Consequently, we believe our results are verifiable in practice.

\begin{acknowledgments}
C. G. and R. R. would like to thank Carlos Vieira and Marcelo Terra Cunha for helpful discussions. We acknowledge the John Templeton Foundation via the Grant Q-CAUSAL No. 61084, the Serrapilheira Institute (Grant No. Serra-1708-15763), the Brazilian National Council for Scientific and Technological Development (CNPq) via the National Institute for Science and Technology on Quantum Information (INCT-IQ), Grants No. 307172/2017-1 and No. 406574/2018-9, the Brazilian agencies MCTIC and MEC, the S\~{a}o Paulo Research Foundation FAPESP (Grant No. 2018/07258-7), and FAEPEX/UNICAMP (Grant No. 3044/19).
\end{acknowledgments}    

\bibliographystyle{apsrev4-2}
\bibliography{bibliography}

\end{document}